\def\draft{n}
\documentclass[12pt]{amsart}
\usepackage[headings]{fullpage}
\usepackage{amssymb,amsmath,amscd,bbm,tikz}
\usepackage{epic,eepic,epsfig}
\usepackage[all]{xy}
\usepackage{url}
\usepackage[bookmarks=true,%
    colorlinks=true,%
    linkcolor=blue,%
    citecolor=blue,%
    filecolor=blue,%
    menucolor=blue,%
    urlcolor=blue,%
    breaklinks=true]{hyperref}
\usepackage{slashed}    

\newtheorem{theorem}{Theorem}[section]
\theoremstyle{definition}
\newtheorem{proposition}[theorem]{Proposition}
\newtheorem{lemma}[theorem]{Lemma}

\newtheorem{remark}[theorem]{Remark}

\newtheorem{example}[theorem]{Example}

\def\printname#1{
        \if\draft y
                \smash{\makebox[0pt]{\hspace{-0.5in}
                        \raisebox{8pt}{\tt\tiny #1}}}
        \fi
}

\newlength{\standardunitlength}
\catcode`\@=11
\long\def\@makecaption#1#2{%
     \vskip 10pt

\setbox\@tempboxa\hbox{
       \small\sf{\bfcaptionfont #1. }\ignorespaces #2}%
     \ifdim \wd\@tempboxa >\captionwidth {%
         \rightskip=\@captionmargin\leftskip=\@captionmargin
         \unhbox\@tempboxa\par}%
       \else
         \hbox to\hsize{\hfil\box\@tempboxa\hfil}%
     \fi}
\font\bfcaptionfont=cmssbx10 scaled \magstephalf
\newdimen\@captionmargin\@captionmargin=2\parindent
\newdimen\captionwidth\captionwidth=\hsize
\catcode`\@=12

\newcommand{\tr}{\operatorname{Tr}}

\def\lbl#1{\label{#1}\printname{#1}}

\def\BN{\mathbbm N}
\def\BZ{\mathbbm Z}
\def\IP{\mathbbm P}
\def\IF{\mathbbm F}
\def\BQ{\mathbbm Q}
\def\BR{\mathbbm R}

\def\CO{\mathcal{O}}
\def\CB{\mathcal{B}}

\def\tq{\tilde{q}}

\newcommand{\fad}{\operatorname{\Phi}_{\mathsf{b}}}
\newcommand{\fadi}{\fad^*}
\newcommand{\fadn}{\phi_{\mathsf{b}}}
\newcommand{\mypsi}[2]{\operatorname{\Psi}_{#1,#2}}

\newcommand{\myh}{\mathsf{h}_{\mathsf{b}}}
\newcommand{\myset}{\Delta_{\mathsf{b}}}

\newcommand{\be}{\begin{equation}}
\newcommand{\ee}{\end{equation}}
\newcommand{\ba}{\begin{aligned}}
\newcommand{\ea}{\end{aligned}}

\newcommand{\re}{{\rm e}}
\newcommand{\ri}{{\rm i}}
\newcommand{\rd}{{\rm d}}
\newcommand{\im}{\mathsf{i}}

\begin{document}
\title[Operators from mirror curves]{Operators from mirror curves and the quantum dilogarithm}
\author{Rinat Kashaev}
\address{Section de Math\'ematiques, Universit\'e de Gen\`eve \\
2-4 rue du Li\`evre, Case Postale 64, 1211 Gen\`eve 4, Switzerland \newline
         {\tt \url{http://www.unige.ch/math/folks/kashaev}}}
\email{Rinat.Kashaev@unige.ch}
\author{Marcos Mari\~no}
\address{Section de Math\'ematiques et D\'epartement de Physique Th\'eorique\\
Universit\'e de Gen\`eve, 1211 Gen\`eve 4, Switzerland }
\email{marcos.marino@unige.ch}
\thanks{%
R.K. and M.M. are supported in part by the Swiss National Science Foundation, subsidies 200020-141329, 200020-149226, 200021-156995, 
and by the NCCR 51NF40-141869 ``The Mathematics of Physics" (SwissMAP).
}

\date{\today }

\begin{abstract}
Mirror manifolds to toric Calabi--Yau threefolds are encoded in algebraic curves. The quantization of these 
curves leads naturally to quantum-mechanical operators on the real line. We show that, for a large number of 
local del Pezzo Calabi--Yau threefolds, these operators are of trace class. In some simple geometries, like local $\IP^2$, we calculate the integral kernel 
of the corresponding operators in terms of Faddeev's quantum dilogarithm. 
Their spectral traces are expressed in terms of multi-dimensional integrals, similar to the state-integrals appearing 
in three-manifold topology, and we show that they can be evaluated explicitly in some cases. Our results provide further verifications of a recent conjecture 
which gives an explicit expression for the Fredholm determinant of these operators, in terms of enumerative invariants of the underlying Calabi--Yau threefolds.   
\end{abstract}

\maketitle

\tableofcontents


\section{Introduction}
\lbl{sec.intro}


The problem of ``quantizing" in a suitable way an algebraic curve appears in many different 
contexts in mathematical physics. For example, the quantization of spectral curves associated 
to integrable systems leads to the Baxter equation. Solving the spectral problem for the Baxter equation, 
after imposing appropriate boundary conditions, gives the solution of the original eigenvalue problem. 
In other cases, the classical limit 
of a quantum theory is encoded in an algebraic curve, and a suitable quantization of the curve makes it possible to go 
beyond the classical limit and to obtain the quantum corrections. 

One important example of such a situation is local mirror symmetry \cite{KKV,CKYZ}. The mirror manifold to 
a toric Calabi--Yau manifold $X$ reduces to an algebraic curve in the exponentiated variables $x$, $y$, of the form 
\be
\label{mc}
W(\re^x, \re^y)=0. 
\ee
The genus zero Gromov--Witten invariants of $X$ can be obtained by computing periods of this curve. 
In \cite{ADKMV} it was pointed out that the higher genus invariants (which can be 
regarded as ``quantum corrections") might be obtained by quantizing in an appropriate way the mirror curve. Building on work 
on supersymmetric gauge theories and quantum integrable systems \cite{NS}, it was shown in \cite{ACDKV} that mirror curves can be 
formally quantized by using the WKB approximation (this method was previously used in \cite{MM}, in the context of Seiberg--Witten curves). 
The quantum corrections obtained in this way 
do not correspond to the conventional higher genus Gromov--Witten invariants, but to a peculiar combination of the 
``refined" BPS invariants of the toric Calabi--Yau manifold \cite{IKV,CKK,NO}. 

The quantization prescription of \cite{ACDKV} leads to ``quantum" periods which are formal WKB series in powers of $\hbar$. 
However, in \cite{KM} it was pointed out that one can associate operators with a
well-defined, discrete spectrum, to the mirror curves of toric Calabi--Yau threefolds. 
The WKB series considered in \cite{ACDKV} turns out to be insufficient to determine this spectrum, 
as it misses non-perturbative corrections of the instanton type. It was suggested 
in \cite{KM} that these corrections would involve the conventional higher genus Gromov--Witten invariants of $X$, 
and this suggestion was verified in some examples in \cite{KM,HW}. These observations were 
deepened and put on a firmer ground in \cite{GHM}, where it was conjectured that, given a mirror curve to a toric 
Calabi--Yau threefold $X$, one can associate to it a trace class operator. Furthermore, \cite{GHM} proposed 
an exact formula for the Fredholm determinant of this trace class operator in terms of the enumerative geometry of $X$. Both 
conjectures were tested in detail in a number of examples, both analytically and numerically.  
The proposal of \cite{GHM} focused on the case in which the mirror curve is of genus one, i.e. in the case in which 
$X$ is a local del Pezzo Calabi--Yau threefold, but it can be suitably generalized to the case of higher genus curves.

In this paper we will perform a detailed study of the operators associated to mirror curves by the procedure 
explained in \cite{GHM}. 
We will prove that, indeed, for a large number of local del Pezzo threefolds, they lead to 
positive-definite, trace class operators on $L^2(\BR)$. Therefore, their spectrum is discrete and positive (as 
it was shown numerically in some examples in \cite{KM,HW,GHM}), and their Fredholm determinants are well-defined. 
One of the most interesting aspects of our study is that 
the operators obtained in this way are closely related to those appearing in 
the quantization of Teichm\"uller theory \cite{K1,K2, K3, K4, FC}. In particular, in the simple case of three-term operators, 
we compute explicitly their integral kernels, which involve 
in an essential way Faddeev's quantum dilogarithm function \cite{Faddeev,FK-QDL}. The spectral traces 
of these operators lead then to multi-dimensional integrals 
that are formally very similar to the state-integral invariants of three-manifolds 
studied recently in \cite{Hi,DGLZ,AK,KLV,DG,AK:complex,Dim}. We use recent techniques 
for the evaluation of state integrals \cite{GK} to calculate these spectral traces in various cases, and we verify that the results 
fully agree with the predictions of \cite{GHM}. 

This paper is organized as follows. In section 2 we give a brief review of the construction of the operators 
from mirror curves, and after introducing some of the necessary ingredients, 
we prove that the relevant operators are of trace class. In addition, we calculate the 
explicit expression of their integral kernels in the simple case of three-term operators. In section 3 we write down 
formulae for the spectral traces of the three-term operators, and we evaluate them explicitly in some cases. In section 4 we compare our 
results to the conjecture of \cite{GHM} for the 
Fredholm determinant of these operators. Finally, in section 5 we conclude and list some interesting open problems.

\subsection*{Acknowledgments}
We would like to thank Andrea Brini, Alba Grassi, Jie Gu, Yasuyuki Hatsuda, Albrecht Klemm, Jonas Reuter, Shamil Shakirov 
and Leon Takhtajan for useful conversations and correspondence. 


\section{Trace class operators and mirror curves}
\lbl{sec.int.quasi}

\subsection{Mirror curves}

The operators which we will study in this paper arise by quantizing mirror curves (in an appropriate sense). We will focus for simplicity on toric (almost) 
del Pezzo Calabi--Yau threefolds, which are 
defined as the total space of the canonical line bundle on a toric (almost) del Pezzo surface $S$,
\be
\label{dP}
X=\CO(K_S) \rightarrow S. 
\ee
They are sometimes called ``local $S$," so for example if $S= \IP^2$, the total space of its canonical line bundle will be called local $\IP^2$. Examples 
of toric (almost) del Pezzos include, besides $\IP^2$, the  Hirzebruch surfaces $\IF_n$, $n=0, 1,2$, and the blowups of $\IP^2$ at $n$ points, denoted by 
${\mathcal B}_n$, for $n=1,2,3$ (note that $\IF_1= {\mathcal B}_1$, and that $\IF_0= \IP^1 \times \IP^1$). 

By standard results in toric geometry (see for example \cite{HKP,CR}), 
toric, almost del Pezzo surfaces can be classified by reflexive polyhedra in two dimensions. The polyhedron $\Delta_S$ associated to a surface $S$ is 
the convex hull of a set of two-dimensional 
vectors 
\be
\label{vector}
\nu^{(i)}=\left(\nu^{(i)}_1, \nu^{(i)}_2\right), \qquad i=1, \cdots, k+2, 
\ee
together with the origin. In order to construct the total space of the canonical line bundle over $S$, we have to consider the extended vectors 
\be
\ba
\overline \nu^{(0)}&=(1, 0,0), \\
\overline \nu^{(i)}&=\left(1, \nu^{(i)}_1, \nu^{(i)}_2\right),\qquad i=1, \cdots, k+2. 
\ea
\ee
They satisfy the relations 
\be
\sum_{i=0}^{k+2} Q^\alpha_i \overline \nu^{(i)}=0, 
\ee
where $Q^\alpha_i$ is a matrix of integers (called the charge matrix) which characterizes the geometry. 

The construction of the mirror geometry to (\ref{dP}) goes back to Batyrev, and it has been recently reviewed in \cite{CR}, 
to which we refer for further details. 
In order to write down the equation for the mirror curve 
to the Calabi--Yau (\ref{dP}), we note that it depends on $k$ complex moduli $z_\alpha$, $\alpha=1, \cdots, k$, which can be parametrized in many possible ways. The most 
useful parametrization involves a modulus $\tilde u$ and a set of ``mass" parameters $\zeta_i$, $i=1, \cdots, r$, where $r$ depends on the geometry 
under consideration \cite{HKP,HKRS}. In terms of 
these variables, the mirror curve for a local del Pezzo Calabi--Yau threefold can be written as, 
\be
\label{ex-W}
W(\re^x, \re^y)= \CO_S(x,y)+ \tilde u=0,  
\ee
where
\be
\label{coxp}
 \CO_S (x,y)=\sum_{i=1}^{k+2} \exp\left( \nu^{(i)}_1 x+  \nu^{(i)}_2 y + f_i(\zeta_j) \right), 
 \ee
and $f_i(\zeta_j)$ are suitable functions of the parameters $\zeta_j$. 

\begin{example} The simplest case of a local del Pezzo is local $\IP^2$. In this case, we have $k=1$. 
The vectors (\ref{vector}) are given by 
\be
\nu^{(1)}=(1,0), \qquad \nu^{(2)}=(0,1), \qquad \nu^{(3)}=(-1,-1). 
\ee
In this geometry there is one complex deformation parameter $\tilde u$, and the function $\CO_{\IP^2}$ is given by  
\be
\label{lp2}
\CO_{\IP^2} \left(x, y \right)= \re^{ x} + \re^{y} + \re^{- x- y}. 
\ee
\end{example}

\begin{example} 
The previous example can be generalized by considering the canonical line bundle over the 
weighted projective space $\IP (1, m, n)$, where $m, n \in \BZ_{>0}$. This is not a smooth manifold, but 
it can be analyzed by using extensions of Gromov--Witten theory, see for example \cite{BC} for a study of the case $n=1$. The vectors are in this case 
\be
\nu^{(1)}=(1,0), \qquad \nu^{(2)}=(0,1), \qquad \nu^{(3)}=(-m,-n), 
\ee
and the function $\CO$ appearing in the mirror curve (\ref{ex-W}) is given by 
\be
\label{lpnm}
\CO_{m,n} \left(x, y \right)= \re^{ x} + \re^{y} + \re^{- m x- n y}. 
\ee
Some of these geometries can arise as degeneration limits of toric del Pezzos. For example, the 
mirror curve to local $\IF_2$ is characterized by the function 
\be
\CO_{\IF_2} \left(x, y \right)= \re^{ x} + \re^{y} + \re^{- 2 x- y} + \zeta \re^{-x}, 
\ee
and when $\zeta=0$ we recover the geometry (\ref{lpnm}) with $m=2$ and $n=1$.
\end{example} 

Some examples of functions  obtained from mirror curves of local del Pezzos can be found in table \ref{table-ops}. Details on the corresponding 
geometries can be found in for example \cite{HKP}. 

\begin{table}
\centering
\begin{tabular}{||  l || l || l || l ||}
\hline
$S$  & $\CO_S (x,y)$  \\ \hline\hline
$\IP^2$ &  $\re^{x}+ \re^y + \re^{-x-y}$  \\ \hline \hline
$\IF_0$ &  $\re^x+ \zeta \re^{-x} + \re^y + \re^{-y}$  \\ \hline \hline
$\IF_1$ &  $\re^x+ \zeta \re^{-x} + \re^y +  \re^{-x -y}$  \\ \hline \hline
$\IF_2$ & $\re^x+ \zeta \re^{-x} + \re^y + \re^{ -2x-y }$ \\ \hline \hline
$\CB_2$ & $ \re^{x} +  \re^{y } +  \re^{-x-y} +\zeta_1 \re^{-y}+ \zeta_2 \re^{-x}$  \\ \hline \hline
$\CB_3$ & $ \re^x +  \re^y+ \re^{-x-y} + \zeta_1 \re^{-x} +\zeta_2 \re^{-y}+ \zeta_3 \re^{x+ y}$ \\ \hline 
\end{tabular}
\caption{The functions $\CO_S(x,y)$ associated to some local del Pezzo Calabi--Yaus.}
\label{table-ops}
\end{table}

\subsection{Quantization} The ``quantization" of the mirror curve (\ref{ex-W}), in the case of local del Pezzos, is based 
on the promotion of the function $\CO_S(x,y)$ to an operator, which 
will be denoted by $\mathsf{O}_S$. This is achieved by simply promoting $x$, $y$ to self-adjoint Heisenberg 
operators $\mathsf{x}$, $\mathsf{y}$ satisfying the commutation relation 
\be
[\mathsf{x}, \mathsf{y}]=\im\hbar. 
\ee
Possible ordering ambiguities are resolved by using Weyl's prescription. As noted in \cite{GHM}, instead 
of studying $\mathsf{O}_S$ (which is not of trace class), one should rather consider its inverse 
\be
\rho_S=\mathsf{O}^{-1}_S. 
\ee
One of our goals in this paper is to show that, for a large number of choices of $S$, this operator exists and is of trace class. 

It will be useful to introduce normalized Heisenberg operators, $\mathsf{p}$ and  $\mathsf{q}$, satisfying the commutation relation
 \begin{equation}
[\mathsf{p},\mathsf{q}]=(2\pi\im)^{-1}.
\end{equation}
The ``coordinate representation'' is given by a realisation in the Hilbert space $L^2(\mathbb{R})$ by the formulae
\begin{equation}
\langle x|\mathsf{q}=x\langle x|,\quad \langle x|\mathsf{p}=\frac1{2\pi\im}\frac{\partial}{\partial x}\langle x|,\quad \langle x|y\rangle=\delta(x-y),\quad {\bf 1}=\int_{\mathbb{R}}|x\rangle \operatorname{d}\!x\langle x|,
\end{equation}
while the ``momentum representation'' is given by 
\begin{equation}
( x|\mathsf{p}=x(x|,\quad (x|\mathsf{q}=\frac\im{2\pi}\frac{\partial}{\partial x}( x|,\quad ( x|y)=\delta(x-y),\quad {\bf 1}=\int_{\mathbb{R}}|x) \operatorname{d}\!x(x|,
\end{equation}
and the transition between these two representations is given by the Fourier kernel
\begin{equation}
\langle x|y)=\re^{2\pi\im xy}.
\end{equation}

The following result is elementary (see also \cite{Si}). 
\begin{lemma}\label{hilb-schm}
 For any $f,g\in L^2(\mathbb{R})$, the operator
 \begin{equation}
\mathsf{G}\equiv f(\mathsf{q})g(\mathsf{p})
\end{equation}
is a Hilbert--Schmidt operator.
\end{lemma}
\begin{proof}
 By using the integral kernel in the mixed basis 
\begin{equation}
 \langle x|\mathsf{G}|y)=f(x)g(y)\langle x|y),
 \end{equation}
we have
\begin{multline}
\operatorname{Tr}(\mathsf{G}^*\mathsf{G})=\int_{\mathbb{R}^2}(y|\mathsf{G}^*|x\rangle\langle x|\mathsf{G}|y)\operatorname{d}\!x\operatorname{d}\!y=
\int_{\mathbb{R}^2}\left|\langle x|\mathsf{G}|y)\right|^2\operatorname{d}\!x\operatorname{d}\!y\\
=
\int_{\mathbb{R}^2}\left|f(x)g(y)\right|^2\operatorname{d}\!x\operatorname{d}\!y
=\|f\|^2\|g\|^2<\infty
\end{multline}
where we use the $L^2$-norm
\begin{equation}
\| f\|^2\equiv\int_{\mathbb{R}}\left|f(x)\right|^2\operatorname{d}\!x.
\end{equation}
\end{proof}

In order to study the properties of the operators associated to local del Pezzo geometries, we will proceed in two steps. First, we will consider the three-term operators 
$\mathsf{O}_{m,n}$, obtained 
by quantization of the function (\ref{lpnm}), and we will establish that their inverse operators
\be
\label{rhoO}
 \rho_{m,n}\equiv \mathsf{O}^{-1}_{m,n}
 \ee
 exist and are trace class operators on $L^2(\BR)$. 
The operators associated to more general del Pezzos can be regarded as perturbations of three-term operators, and this makes it possible to 
show that they are as well of trace class. The operator associated to local $\IF_0$ is somewhat special, but it can be analyzed with similar techniques, as we will see. 

A key ingredient to study the operators $\rho_{m,n}$ is an explicit determination of their integral kernels, which involves Faddeev's quantum dilogarithm $\fad(x)$ 
\cite{Faddeev,FK-QDL}. A summary of the properties of this function can be found in the Appendix. Here, 
we list some preliminary results for the analysis of the three-term operators. 
We fix a positive real number $\mathsf{b}$ and define a set
\begin{equation}\label{myset}
\myset\equiv\left\{ (a,c)\in\mathbb{R}_{>0}^2\left|\ a+c<\myh\equiv\frac{\mathsf{b}+\mathsf{b}^{-1}}2\right.\right\}.
\end{equation}
We will also denote
\be
c_{\mathsf{b}}\equiv\im \myh. 
\ee
Define a function
\begin{equation}\label{mypsi}
\mypsi{a}{c}(x)\equiv \frac{\re^{2\pi ax}}{\fad(x-\im(a+c))},\quad x\in\mathbb{R},\ (a,c)\in\myset,
\end{equation}
which is a nowhere vanishing Schwartz function in variable $x$, i.e. a smooth and rapidly decreasing function at infinity. 
Indeed, the conditions on parameters $a$ and $c$ are such that $\mypsi{a}{c}(x)$ is a restriction of a meromorphic 
function in the complex plane whose poles and zeros do not belong to the real axis $\mathbb{R}\subset\mathbb{C}$, and
the formula
\begin{equation}
\lim_{x\to-\infty}\fad(x+\im y)=1, \quad \forall y\in\mathbb{R},
\end{equation}
implies that
\begin{equation}
\lim_{x\to-\infty}(\Psi_{a,c}(x)\re^{-2\pi a x}-1)=0.
\end{equation}
By using the equalities
\begin{equation}
|\mypsi{a}{c}(x)|=\left|\fad(-x+\im(a+c))\re^{2\pi ax-\pi\im(x-\im(a+c))^2}\right|=|\fad(-x+\im(a+c))|\re^{-2\pi cx},
\end{equation}
we also have
\begin{equation}
\lim_{x\to+\infty}(|\Psi_{a,c}(x)| \re^{2\pi c x}-1)=0.
\end{equation}

\begin{lemma}
 For any $(a_1,a_2)\in\myset$, one has
 \begin{equation}\label{norm-psi}
\|\mypsi{a_1}{a_2}\|^2=\prod_{j=1}^3|\fad(c_{\mathsf{b}}-2\im a_j)|,\quad a_3\equiv \myh-a_1-a_2.
\end{equation}
\end{lemma}
\begin{proof}
The formula follows  from the integral Ramanujan formula
 \be
\varphi(z,w)\equiv\int_{\mathbb{R}}\frac{\fad\!\left(t+\frac z 2\right)}{\fad\!\left(t-\frac z 2\right)}\re^{2\pi \im tw}\operatorname{d}\!t 
=\fadn\!\left(z-c_{\mathsf{b}}\right)\fadn\!\left(w+c_{\mathsf{b}}\right)\fadn\!\left(-z-w+c_{\mathsf{b}}\right),
\ee
where the normalized quantum dilogarithm function
\begin{equation}\label{fadn}
\fadn(z)\equiv \frac{\fad(z)}{\fad(0)}\re^{-\pi\im z^2/2}
\end{equation}
has the properties
\begin{equation}
\fadn(z)\fadn(-z)=1, \quad \overline{\fadn(z)}=\fadn(-\bar z),\quad \frac{\fadn(z-\im\mathsf{b}/2)}{\fadn(z+\im\mathsf{b}/2)}=2\cosh(\pi\mathsf{b}z).
\end{equation}

 We have
\begin{multline}
\|\mypsi{a_1}{a_2}\|^2=\int_{\mathbb{R}}\frac{\fad(t+\im(a_1+a_2))}{\fad(t-\im(a_1+a_2))} \re^{4\pi a_1 t}\operatorname{d}\!t=\varphi(2\im(a_1+a_2),-2\im a_1)\\=\varphi(2\im(\myh-a_3),-2\im a_1)
=\fadn\!\left(c_{\mathsf{b}}-2\im a_3\right)\fadn\!\left(c_{\mathsf{b}}-2\im a_1\right)\fadn\!\left(c_{\mathsf{b}}-2\im a_2\right),
\end{multline}
which is exactly \eqref{norm-psi}, as we have the equality
\begin{equation}
|\fad(z)|=|\fadn(z)|\quad \text{if}\quad \Im(z)\Re(z)=0.
\end{equation}
\end{proof}
\begin{remark}\label{rem1}
More generally, we have the following Fourier transformation formula
 \begin{multline}
\int_{\mathbb{R}}|\mypsi{a_1}{a_2}(t)|^2\re^{2\pi\im xt}\operatorname{d}\!t
=\int_{\mathbb{R}}\frac{\fad(t+\im(a_1+a_2))}{\fad(t-\im(a_1+a_2))} \re^{2\pi\im t(x-2\im a_1)}\operatorname{d}\!t\\
=\varphi(2\im(a_1+a_2),x-2\im a_1)\\=
\fadn\!\left(2\im(a_1+a_2)-c_{\mathsf{b}}\right)\fadn\!\left(x-2\im a_1+c_{\mathsf{b}}\right)\fadn\!\left(-2\im(a_1+a_2)-x+2\im a_1+c_{\mathsf{b}}\right)\\
=
\fadn\!\left(c_{\mathsf{b}}-2\im a_3\right)\fadn\!\left(x-2\im a_1+c_{\mathsf{b}}\right)\fadn\!\left(-x-2\im a_2+c_{\mathsf{b}}\right)\\
=
\fadn\!\left(c_{\mathsf{b}}-2\im a_3\right)\frac{\fadn\!\left(x-2\im a_1+c_{\mathsf{b}}\right)}{\fadn\!\left(x+2\im a_2-c_{\mathsf{b}}\right)}.
\end{multline}

\end{remark}
\begin{lemma} \label{fad2} One has the following equalities
\begin{equation}
\fad(\mathsf{p})\re^{2\pi\mathsf{b}\mathsf{q}}\fadi(\mathsf{p})=\re^{2\pi\mathsf{b}\mathsf{q}}+\re^{2\pi\mathsf{b}(\mathsf{p}+\mathsf{q})},
\end{equation}

 \begin{equation}
\fadi(\mathsf{q})\fad(\mathsf{p})\re^{2\pi\mathsf{b}\mathsf{q}}\fadi(\mathsf{p})\fad(\mathsf{q})=\re^{2\pi\mathsf{b}\mathsf{q}} +\re^{2\pi\mathsf{b}(\mathsf{p}+\mathsf{q})} 
+\re^{2\pi\mathsf{b}(\mathsf{p}+2\mathsf{q})}
\end{equation}
\end{lemma}
\begin{proof} We have
\begin{equation}
\fad(\mathsf{p})\re^{2\pi\mathsf{b}\mathsf{q}}\fadi(\mathsf{p})=
\re^{\pi\mathsf{b}\mathsf{q}}\frac{\fad(\mathsf{p}-\im\mathsf{b}/2)}{\fad(\mathsf{p}+\im\mathsf{b}/2)}\re^{\pi\mathsf{b}\mathsf{q}}=
\re^{\pi\mathsf{b}\mathsf{q}}\left(1+\re^{2\pi\mathsf{b}\mathsf{p}}\right)\re^{\pi\mathsf{b}\mathsf{q}}
  =\re^{2\pi\mathsf{b}\mathsf{q}}+\re^{2\pi\mathsf{b}(\mathsf{p}+\mathsf{q})}
\end{equation}
and, by using the previous formula,
\begin{multline}
 \fadi(\mathsf{q})\fad(\mathsf{p})\re^{2\pi\mathsf{b}\mathsf{q}}\fadi(\mathsf{p})\fad(\mathsf{q})=
  \fadi(\mathsf{q})\left(\re^{2\pi\mathsf{b}\mathsf{q}}+\re^{2\pi\mathsf{b}(\mathsf{p}+\mathsf{q})}\right)\fad(\mathsf{q})\\
  =
\re^{2\pi\mathsf{b}\mathsf{q}}+  \re^{\pi\mathsf{b}(\mathsf{p}+\mathsf{q})}\frac{\fad(\mathsf{q}-\im\mathsf{b}/2)}{\fad(\mathsf{q}+\im\mathsf{b}/2)}\re^{\pi\mathsf{b}(\mathsf{p}+\mathsf{q})}
 =
\re^{2\pi\mathsf{b}\mathsf{q}}+  \re^{\pi\mathsf{b}(\mathsf{p}+\mathsf{q})}\left(1+\re^{2\pi\mathsf{b}\mathsf{q}}\right)\re^{\pi\mathsf{b}(\mathsf{p}+\mathsf{q})}
\\
=\re^{2\pi\mathsf{b}\mathsf{q}} +\re^{2\pi\mathsf{b}(\mathsf{p}+\mathsf{q})} 
+\re^{2\pi\mathsf{b}(\mathsf{p}+2\mathsf{q})}.
\end{multline}
\end{proof}
 
 \subsection{Generic three-term operator}
Consider the operator associated to the function (\ref{lpnm}):
 \begin{equation}
\mathsf{O}_{m,n}=\re^{\mathsf{x}}+\re^{\mathsf{y}} +\re^{-m\mathsf{x}-n\mathsf{y}},\quad m,n\in\mathbb{R}_{>0}.
\end{equation}
Note that $m, n$ can be {\it a priori} arbitrary positive, real numbers, although in the operators arising from the mirror curves they are integers. By using Lemma~\ref{fad2} and the substitutions
\begin{equation}
\mathsf{x}\equiv 2\pi\mathsf{b}\frac{(n+1)\mathsf{p}+n\mathsf{q}}{m+n+1},\quad \mathsf{y}\equiv -2\pi\mathsf{b}\frac{m\mathsf{p}+(m+1)\mathsf{q}}{m+n+1}, 
\end{equation}
so that
\be
\label{b-hbar}
\hbar=\frac{2\pi\mathsf{b}^2}{m+n+1}, 
\ee
we have
 \begin{multline}
\re^{-\mathsf{y}/2} \mathsf{O}_{m,n}\re^{-\mathsf{y}/2}=\re^{\mathsf{x}-\mathsf{y}} +1+\re^{-m\mathsf{x}-(n+1)\mathsf{y}} 
=\re^{2\pi\mathsf{b}(\mathsf{p}+\mathsf{q})} +1+\re^{2\pi\mathsf{b}\mathsf{q}} 
\\
=1+\fad(\mathsf{p})\re^{2\pi\mathsf{b}\mathsf{q}}\fadi(\mathsf{p})
=\fad(\mathsf{p})\frac{\fad(\mathsf{q}-\im\mathsf{b}/2)}{\fad(\mathsf{q}+\im\mathsf{b}/2)}\fadi(\mathsf{p}).
\end{multline}
By defining an operator
\begin{equation}
\mathsf{A}_{m,n}\equiv \fadi(\mathsf{q}-\im\mathsf{b}/2)\fadi(\mathsf{p})\re^{\frac{\pi\mathsf{b}(m+1)}{m+n+1}\mathsf{q}}\re^{\frac{\pi\mathsf{b}m}{m+n+1}\mathsf{p}},
\end{equation}
we obtain the following formula for the inverse operator (\ref{rhoO}), 
\begin{equation}
\label{oaa}
\rho_{m,n}=\mathsf{A}_{m,n}^*\mathsf{A}_{m,n}.
\end{equation}
Let us now rewrite $\mathsf{A}_{m,n}$ in the form 
\be
\label{app}
 \mathsf{A}_{m,n}=\fadi\!\left(\mathsf{q}-\im\frac{\mathsf{b}}{2}\right)\re^{\frac{\pi\mathsf{b}(m+1)}{m+n+1}\mathsf{q}}
 \fadi\!\left(\mathsf{p}-\im\frac{\mathsf{b}(m+1)}{2(m+n+1)}\right)\re^{\frac{\pi\mathsf{b}m}{m+n+1}\mathsf{p}}
 =\mypsi{a+c}{cn}(\mathsf{q})\mypsi{a}{c}(\mathsf{p}),
\ee
 where
 \be
 a\equiv\frac{\mathsf{b}m}{2(m+n+1)},\quad c\equiv\frac{\mathsf{b}}{2(m+n+1)}.
 \ee
 We see that $\mathsf{A}_{m,n}$ is a Hilbert--Schmidt operator due to the inclusions
 \begin{equation}
\left(\frac{\mathsf{b}(m+1)}{2(m+n+1)},\frac{\mathsf{b}n}{2(m+n+1)}\right)\in\myset\ni\left(\frac{\mathsf{b}m}{2(m+n+1)},\frac{\mathsf{b}}{2(m+n+1)}\right)
\end{equation}
and Lemma~\ref{hilb-schm}.
\begin{theorem} The operator $\rho_{m,n}$ is positive-definite and of trace class. 
\end{theorem}
\begin{proof}
Due to (\ref{oaa}), we have that, for any $f\in L^2(\mathbb{R})$, 
\be
\int_{\mathbb{R}}f^*(x) (\rho_{m,n} f)(x) \rd x = \Vert \mathsf{A}_{m,n} f\Vert^2, 
\ee
and since $\mathsf{A}_{m,n}$ is invertible, we conclude that $\rho_{m,n}$ is positive-definite. Since the product of two Hilbert--Schmidt operators is trace class, and $\mathsf{A}_{m,n}$ is 
Hilbert--Schmidt, we also conclude that $\rho_{m,n}$ is of trace class. 
\end{proof}

\begin{proposition}
The integral kernel of $\rho_{m,n}$ in the momentum representation is given by the formula
\begin{equation}
\label{ex-k}
(x|\rho_{m,n} |y)=\frac{\overline{\mypsi{a}{c}(x)}\mypsi{a}{c}(y)}{2\mathsf{b}\cosh\left(\pi\frac{x-y+\im (a+c-nc)}{\mathsf{b}}\right)}.
\end{equation}
\end{proposition}
\begin{proof} From (\ref{oaa}) and (\ref{app}) we obtain
\be
\rho_{m,n}=\mypsi{a}{c}^*(\mathsf{p})\mypsi{a+c}{cn}^*(\mathsf{q}) \mypsi{a+c}{cn}(\mathsf{q})\mypsi{a}{c}(\mathsf{p})
=\mypsi{a}{c}^*(\mathsf{p})\left|\mypsi{a+c}{cn}(\mathsf{q})\right|^2\mypsi{a}{c}(\mathsf{p}).
\ee
By using the difference functional equation for the quantum dilogarithm and the Fourier integral
\begin{equation}
\frac1{\cosh(\pi x)}= \int_{\mathbb{R}}\frac{\re^{2\pi\im xy}}{\cosh(\pi y)}\operatorname{d}\!y,
\end{equation}
 and denoting 
 \begin{equation}\label{heq}
h\equiv a+c-nc,
\end{equation}
 we have
\begin{multline}
\label{psi-cosh}
 \left|\mypsi{a+c}{cn}(x)\right|^2=\left|\frac{\re^{2\pi (a+c)x}}{\fad(x-\im \mathsf{b}/2)}\right|^2=
 \frac{\fad(x+\im \mathsf{b}/2)}{\fad(x-\im  \mathsf{b}/2)}\re^{4\pi (a+c)x}\\
 =\frac{\re^{4\pi (a+c)x}}{1+\re^{2\pi \mathsf{b} x}}=\frac{\re^{2\pi hx}}{2\cosh(\pi \mathsf{b} x)}=
 \int_{\mathbb{R}}\frac{\re^{2\pi hx} \re^{2\pi\im \mathsf{b} xy}}{2\cosh(\pi y)}\operatorname{d}\!y
 =
 \int_{\mathbb{R}}\frac{\re^{2\pi\im \mathsf{b} x(y-\im h/\mathsf{b}) }}{2\cosh(\pi y)}\operatorname{d}\!y\\
  =
 \int_{\mathbb{R}-\im h/\mathsf{b}}\frac{\re^{2\pi\im \mathsf{b} xz }}{2\cosh(\pi (z+\im h/\mathsf{b}))}\operatorname{d}\!z
 =
 \int_{\mathbb{R}}\frac{\re^{2\pi\im  xt }}{2\mathsf{b}\cosh(\pi (t+\im h)/\mathsf{b}))}\operatorname{d}\!t,
\end{multline}
where in the last equality we have shifted the line  of integration by using analyticity of the integrand in the strip $|\Im z|\le |h|/\mathsf{b}$. Finally, we remark that if we have a Fourier integral of the form
\begin{equation}
f(x)= \int_{\mathbb{R}}\tilde f(t)\re^{2\pi\im  xt }\operatorname{d}\!t,
\end{equation}
then we have
\begin{multline}
 (x|f(\mathsf{q})|y)= \int_{\mathbb{R}}\tilde f(t)(x|\re^{2\pi\im  \mathsf{q}t}|y)\operatorname{d}\!t
 = \int_{\mathbb{R}}\tilde f(t)(x-t|y)\operatorname{d}\!t\\
 = \int_{\mathbb{R}}\tilde f(t)\delta(x-t-y)\operatorname{d}\!t=\tilde f(x-y),
\end{multline}
so that
\begin{multline}
(x|\rho_{m,n} |y)=(x|\mypsi{a}{c}^*(\mathsf{p})\left|\mypsi{a+c}{cn}(\mathsf{q})\right|^2\mypsi{a}{c}(\mathsf{p})|y)\\
=\overline{\mypsi{a}{c}(x)}\left(x\left|\left|\mypsi{a+c}{cn}(\mathsf{q})\right|^2\right|y\right)\mypsi{a}{c}(y)
=\frac{\overline{\mypsi{a}{c}(x)}\mypsi{a}{c}(y)}{2\mathsf{b}\cosh\left(\pi\frac{x-y+\im h}{\mathsf{b}}\right)}.
\end{multline}

\end{proof}

\subsection{Local $\mathbb{F}_0$} The operator associated to the mirror curve of local $\IF_0$ is obtained by quantization of the function $\CO_{\mathbb{F}_0}$ appearing 
in table \ref{table-ops}, and it reads 
 \begin{equation}
\mathsf{O}_{\mathbb{F}_0}=\re^{\mathsf{x}} +\zeta \re^{-\mathsf{x}} +\re^{\mathsf{y}} +\re^{-\mathsf{y}}.
\end{equation}
This operator can not be regarded as a perturbation of a three-term operator of the form $\mathsf{O}_{m,n}$: when $\zeta=0$, the resulting operator is precisely the one associated 
to geodesic lenghts in quantum Teichm\"uller theory \cite{K1,K2,FC}, and it has a continuous spectrum \cite{K3,K4,FT}. We then assume $\zeta>0$. We will now show that, when this is the case, 
the inverse operator $\rho_{\mathbb{F}_0}=\mathsf{O}^{-1}_{\mathbb{F}_0}$ exists 
and is of trace class. Let us set 
\begin{equation}
\mathsf{x}=\pi\mathsf{b}(\mathsf{p}+2\mathsf{q}),\quad \mathsf{y}=\pi\mathsf{b}\mathsf{p},
\end{equation}
so that
\be
\hbar = \pi \mathsf{b}^2. 
\ee
By using Lemma~\ref{fad2} we have
 \begin{multline}
\re^{\mathsf{x}/2} \mathsf{O}_{\mathbb{F}_0}\re^{\mathsf{x}/2}-\zeta =\re^{2\mathsf{x}} +\re^{\mathsf{x}+\mathsf{y}} 
+\re^{\mathsf{x}-\mathsf{y}}=\re^{2\pi\mathsf{b}(\mathsf{p}+2\mathsf{q})} +\re^{2\pi\mathsf{b}(\mathsf{p}+\mathsf{q})} 
+\re^{2\pi\mathsf{b}\mathsf{q}}\\
=\fadi(\mathsf{q})\fad(\mathsf{p})\re^{2\pi\mathsf{b}\mathsf{q}}\fadi(\mathsf{p})\fad(\mathsf{q}).
\end{multline}
Thus,
\begin{multline}
\fadi(\mathsf{p})\fad(\mathsf{q})\re^{\mathsf{x}/2} \mathsf{O}_{\mathbb{F}_0}\re^{\mathsf{x}/2}\fadi(\mathsf{q})\fad(\mathsf{p})=\zeta+\re^{2\pi\mathsf{b}\mathsf{q}}=\zeta 
\left(1+\re^{2\pi\mathsf{b}(\mathsf{q}-\mu)}\right)\\
=\zeta \frac{\fad(\mathsf{q}-\mu-\im\mathsf{b}/2)}{\fad(\mathsf{q}-\mu+\im\mathsf{b}/2)},
\end{multline}
where we have introduced a new parameter $\mu$ through the equation
\begin{equation}
\zeta=\re^{2\pi\mathsf{b}\mu}.
\end{equation}
By defining an operator
\begin{equation}
\mathsf{B}\equiv \fadi(\mathsf{q}-\mu-\im\mathsf{b}/2)\fadi(\mathsf{p})\fad(\mathsf{q})\re^{\pi\mathsf{b}\mathsf{p}/2}\re^{\pi\mathsf{b}\mathsf{q}},
\end{equation}
we obtain the following formula
\begin{equation}
\zeta \mathsf{O}^{-1}_{\mathbb{F}_0}=\mathsf{B}^*\mathsf{B}.
\end{equation}
By rewriting $\mathsf{B}$ in the form 
\begin{multline}
 \mathsf{B}=\fadi(\mathsf{q}-\mu-\im\mathsf{b}/2)\fadi(\mathsf{p})\re^{\pi\mathsf{b}\mathsf{p}/2}\fad(\mathsf{q}+\im\mathsf{b}/4)\re^{\pi\mathsf{b}\mathsf{q}}\\
 =\fadi(\mathsf{q}-\mu-\im\mathsf{b}/2)\re^{2\pi\lambda\mathsf{q}}\fadi(\mathsf{p}-\im\lambda)\re^{\pi\mathsf{b}(\mathsf{p}-\im\lambda)/2}\fad(\mathsf{q}+\im\mathsf{b}/4)\re^{\pi(\mathsf{b}-2\lambda)
 \mathsf{q}}\\ 
 =\re^{\pi\lambda(2\mu-\im\mathsf{b}/2)}\mypsi{\lambda}{\mathsf{b}/2-\lambda}(\mathsf{q}-\mu)
 \mypsi{\mathsf{b}/4}{\lambda-\mathsf{b}/4}(\mathsf{p})
 \mypsi{\mathsf{b}/2-\lambda}{\lambda-\mathsf{b}/4}^*(\mathsf{q}),
 \end{multline}
 we see that, for any $\lambda$ satisfying the inequalities
 \begin{equation}
\mathsf{b}/4<\lambda<\mathsf{b}/2,
\end{equation}
$\mathsf{B}$ is a Hilbert--Schmidt operator by Lemma~\ref{hilb-schm} and the fact that Hilbert--Schmidt operators form a two-sided ideal in the algebra 
of bounded operators. It follows that 
$\rho_{\IF_0}$ is positive-definite and of trace class. 
\subsection{Perturbed operators} With the exception of the operator for local $\IF_0$ that we have just considered, 
the operators appearing in Table \ref{table-ops} are perturbations of the operators $\mathsf{O}_{m,n}$.  
Let us then consider an operator of the form, 
 \begin{equation}
\mathsf{O}=\mathsf{V} +\mathsf{A}^{-1}(\mathsf{A}^*)^{-1},
\end{equation}
where $\mathsf{V}$ is a self-adjoint positive operator, and $\mathsf{A}$ is an invertible Hilbert--Schmidt operator.
We write
\begin{equation}
\mathsf{A}\mathsf{O}\mathsf{A}^*=\mathsf{A}\mathsf{V}\mathsf{A}^*+1,
\end{equation}
or equivalently
\begin{equation}
\mathsf{O}^{-1}=\mathsf{A}^*(\mathsf{A}\mathsf{V}\mathsf{A}^*+1)^{-1}\mathsf{A}=\mathsf{C}^*\mathsf{C},
\end{equation}
where
\begin{equation}
\mathsf{C}\equiv (\mathsf{A}\mathsf{V}\mathsf{A}^*+1)^{-1/2}\mathsf{A}
\end{equation}
 is evidently a Hilbert--Schmidt operator. It follows that $\mathsf{O}^{-1}$ is of trace class and positive-definite. We can now apply this general argument to 
 the operators $\rho_S= \mathsf{O}^{-1}_S$ obtained by quantizing the functions ${\mathcal O}_S$ in table \ref{table-ops}, when $S= {\mathcal B}_{n+1}$, $\IF_n$, $n=1,2$. 
 Provided the perturbation parameters $\zeta_i$ are positive, 
 the operators $\rho_S$ will be of trace class and positive-definite. This argument also applies to other operators associated to local del Pezzo threefolds. 

\subsection{Symmetries and equivalences}
\label{sym-eqs}
There are a number of unitary equivalences in the three parameter family of three-term operators 
\begin{equation}
\left\{\mathsf{O}_{m,n}(\hbar)=\re^\mathsf{x}+\re^\mathsf{y}+\re^{-m\mathsf{x}-n\mathsf{y}}\ \left\vert\ [\mathsf{x},\mathsf{y}]=\im\hbar,\ (m,n,\hbar)\in\mathbb{R}^3_{>0}\right.\right\}.
\end{equation}
 The first set of such equivalences is produced by the action of the order three cyclic group $\mathbb{Z}_3$ which cyclically permutes the three operators
\begin{equation}
\mathsf{x}\mapsto \mathsf{y}\mapsto -m\mathsf{x}-n\mathsf{y}\mapsto\mathsf{x}.
\end{equation}
The corresponding equivalences in the parameter space read
\begin{equation}
\label{z3-eq}
(m,n,\hbar)\sim (n/m,1/m,m\hbar)\sim (1/n,m/n,n\hbar).
\end{equation}

In the case of $m=n=1$, the above equivalence becomes a $\mathbb{Z}_3$ symmetry. To see this explicitly, let us define a unitary operator
 \begin{equation}
 \label{Vdef}
\mathsf{V}\equiv \re^{-\pi\im/3}\re^{2\pi\im\mathsf{p}^2}\re^{\pi\im\mathsf{q}^2}\re^{\pi\im\mathsf{p}^2}
=\re^{-\pi\im/12}\re^{\pi\im\mathsf{p}^2}\mathsf{F},
\end{equation}
where $\mathsf{F}$ is the Fourier operator defined by its 
integral kernels in the ``position'' or ``momentum'' representations
 \begin{equation}
\langle x \vert \mathsf{F}\vert y\rangle=\re^{2\pi\im xy}= (x \vert \mathsf{F}\vert y).
\end{equation}
$\mathsf{V}$ can also be written in other forms
\begin{equation}
\mathsf{V}=\left(\re^{-\pi\im/6}\re^{\pi\im\mathsf{p}^2}\re^{\pi\im\mathsf{q}^2}\right)^2=\re^{\pi\im/6}\re^{-\pi\im\mathsf{q}^2}\re^{-\pi\im\mathsf{p}^2}\mathsf{P},
\end{equation}
where $\mathsf{P}$ is the parity operator:
\begin{equation}
\langle x \vert \mathsf{P}\vert y\rangle=\delta(x+y)= (x \vert \mathsf{P}\vert y).
\end{equation}
Various integral kernels of $\mathsf{V}$ are as follows, 
 \begin{multline}
\langle x \vert \mathsf{V}\vert y\rangle=\re^{-\pi\im/12}\re^{\pi\im( 2x+y)y}= (y \vert \mathsf{V}\vert x),
\\
 \langle x \vert \mathsf{V}\vert y)=\re^{\pi\im/6}\re^{-\pi\im(x+y)^2},\quad
( x \vert \mathsf{V}\vert y\rangle=\re^{-\pi\im/12}\re^{\pi\im x^2}\delta(x-y).
\end{multline}
It is easily verified that $\mathsf{V}$ is of order three and commutes with $\mathsf{O}_{1,1}$. This means that $\rho_{1,1}$ and $\mathsf{V}$ can be diagonalized simultaneously. The $\mathbb{Z}_3$ symmetry of this operator reflects the corresponding symmetry of the underlying Calabi--Yau threefold, local $\IP^2$ (which is in fact the resolution of the 
orbifold $\mathbb{C}^3/\mathbb{Z}_3$).

In addition to the general equivalence (\ref{z3-eq}), some additional equivalences appear when the parameters are restricted by some relations.
\begin{proposition}
The relation in the parameter space 
 \begin{equation}\label{eqrel}
(m,n,\hbar)\sim ((m+n)^2/m,n/m,m\hbar)
\end{equation}
corresponds to a unitary equivalence of the associated operators if either
\begin{equation}\label {eN=11} 
m+n=1,
\end{equation}
or
 \begin{equation}\label{eN=-11}
m+n=2\pi/\hbar,
\end{equation}
or else
\begin{equation}\label{e=1}
m+n=N,\quad \hbar=2\pi M/N, 
\end{equation}
where $M,N\in\mathbb{Z}_{>1}$ with $\operatorname{gcd}(M,N)=1$.
\end{proposition}
\begin{proof}
The three cases correspond to the condition
\begin{equation}
m+n=\mathsf{b}^{\epsilon-1} N,\quad \epsilon\in\{\pm1\},\quad N\in\mathbb{Z}, 
\end{equation}
where $N=1$ in the first two cases with $\epsilon=1$ for \eqref{eN=11} and $\epsilon=-1$ for \eqref{eN=-11}, while in the case \eqref{e=1}, $N$ rests arbitrary but $\epsilon=1$.

 By using the parametrization 
 \begin{equation}
\hbar=2\pi\mathsf{b}^2,\quad \mathsf{x}=2\pi\mathsf{b}\mathsf{q},\quad  \mathsf{y}=2\pi\mathsf{b}\mathsf{p},
\end{equation}
we write
\begin{multline}
\re^{\mathsf{x}}+ \re^{\mathsf{y}}= \re^{\mathsf{x}/2}\left(1+ \re^{\mathsf{y}-\mathsf{x}}\right)\re^{\mathsf{x}/2}=
\re^{\pi\mathsf{b}\mathsf{q}}\frac{\fad\left(\mathsf{p}-\mathsf{q}-\im\mathsf{b}/2\right)}{\fad\left(\mathsf{p}-\mathsf{q}+\im\mathsf{b}/2\right)}\re^{\pi\mathsf{b}\mathsf{q}}\\
=\fad\left(\mathsf{p}-\mathsf{q}\right)\re^{2\pi\mathsf{b}\mathsf{q}}\fad\left(\mathsf{p}-\mathsf{q}\right)^{-1},
\end{multline}
so that 
\begin{multline}
\fad\left(\mathsf{p}-\mathsf{q}\right)^{-1}\mathsf{O}_{m,n}(\hbar)
\fad\left(\mathsf{p}-\mathsf{q}\right)-\re^{2\pi\mathsf{b}\mathsf{q}}\\=\fad\left(\mathsf{p}-\mathsf{q}\right)^{-1}\re^{-2\pi\mathsf{b}(m\mathsf{q}+n\mathsf{p})}
\fad\left(\mathsf{p}-\mathsf{q}\right)\\
=\re^{-\pi\mathsf{b}(m\mathsf{q}+n\mathsf{p})}\frac{\fad\left(\mathsf{p}-\mathsf{q}-\im\mathsf{b}(m+n)/2\right)}{\fad\left(\mathsf{p}-\mathsf{q}+\im\mathsf{b}(m+n)/2\right)}\re^{-\pi\mathsf{b}(m\mathsf{q}+n\mathsf{p})}\\=\re^{-\pi\mathsf{b}(m\mathsf{q}+n\mathsf{p})}\frac{\fad\left(\mathsf{p}-\mathsf{q}-\im\mathsf{b}^\epsilon N/2\right)}{\fad\left(\mathsf{p}-\mathsf{q}+\im\mathsf{b}^\epsilon N/2\right)}\re^{-\pi\mathsf{b}(m\mathsf{q}+n\mathsf{p})}\\=
\re^{-\pi\mathsf{b}(m\mathsf{q}+n\mathsf{p})}\left(-\re^{\pi\im\mathsf{b}^{2\epsilon}(1-N)}\re^{2\pi\mathsf{b}^{\epsilon} (\mathsf{p}-\mathsf{q})};\re^{2\pi\im\mathsf{b}^{2\epsilon}}\right)_N\re^{-\pi\mathsf{b}(m\mathsf{q}+n\mathsf{p})},
\end{multline}
where we use the notation
\begin{equation}
(x;q)_N\equiv\prod_{j=0}^{N-1}(1-xq^j),
\end{equation}
and the functional equations for the quantum dilogarithm
 \begin{equation}
\frac{\fad\left(x-\im\mathsf{b}^{\epsilon}N/2\right)}{\fad\left(x+\im\mathsf{b}^\epsilon N/2\right)}=\left(-\re^{\pi\im\mathsf{b}^{2\epsilon}(1-N)}\re^{2\pi\mathsf{b}^{\epsilon} x};\re^{2\pi\im\mathsf{b}^{2\epsilon}}\right)_N,\quad \forall N\in\mathbb{Z},\quad \epsilon\in\{\pm1\}.
\end{equation}
Now, we remark that
\begin{equation}
\left(-q^{(1-N)/2}x;q\right)_N=x^N+1,
\end{equation}
either if $N=1$ or else if $N>1$ and $q$ is a primitive $N$-th root of unity. The latter is the case if $\mathsf{b}^2=M/N$ with mutually prime integers $M$ and $N$.
Thus, we conclude that
\begin{multline}
\fad\left(\mathsf{p}-\mathsf{q}\right)^{-1}\mathsf{O}_{m,n}(\hbar)
\fad\left(\mathsf{p}-\mathsf{q}\right)\\=\re^{2\pi\mathsf{b}\mathsf{q}}+\re^{-\pi\mathsf{b}(m\mathsf{q}+n\mathsf{p})}\left(\re^{2\pi\mathsf{b}^{\epsilon} N(\mathsf{p}-\mathsf{q})}+1\right)\re^{-\pi\mathsf{b}(m\mathsf{q}+n\mathsf{p})}\\
=\re^{2\pi\mathsf{b}\mathsf{q}}+\re^{2\pi\mathsf{b}(\mathsf{b}^{\epsilon-1} N(\mathsf{p}-\mathsf{q})-m\mathsf{q}-n\mathsf{p})}+\re^{-2\pi\mathsf{b}(m\mathsf{q}+n\mathsf{p})}\\
=\re^{2\pi\mathsf{b}\mathsf{q}}+\re^{2\pi\mathsf{b}((m+n)(\mathsf{p}-\mathsf{q})-m\mathsf{q}-n\mathsf{p})}+\re^{-2\pi\mathsf{b}(m\mathsf{q}+n\mathsf{p})}\\=\re^{2\pi\mathsf{b}\mathsf{q}}+\re^{2\pi\mathsf{b}(m\mathsf{p}-(2m+n)\mathsf{q})}+\re^{-2\pi\mathsf{b}(m\mathsf{q}+n\mathsf{p})}=\re^{\mathsf{x}'}+ \re^{\mathsf{y'}}+\re^{-m'\mathsf{x}'-n'\mathsf{y}'},
\end{multline}
where
\begin{equation}
\mathsf{x}'\equiv 2\pi\mathsf{b}\mathsf{q}=\mathsf{x},\quad \mathsf{y}'\equiv 2\pi\mathsf{b}(m\mathsf{p}-(2m+n)\mathsf{q})=m\mathsf{y}-(2m+n)\mathsf{x}
\end{equation}
so that 
\begin{equation}
[\mathsf{x'},\mathsf{y'}]=\hbar',\quad \hbar'\equiv m\hbar,
\end{equation}
and 
\begin{equation}
m'\equiv (m+n)^2/m,\quad n'\equiv n/m.
\end{equation}

\end{proof}


\section{Calculation of spectral traces}

Since the operator $\rho_{m,n}$ is of trace class, all its spectral traces 
\be
\tr \rho_{m,n}^L, \qquad L \in \BZ_{>0}
\ee
exist. They can be computed by using the explicit expression for the kernel (\ref{ex-k}), which leads to an $L$-dimensional integral involving products of quantum 
dilogarithms. These integrals are then similar to state-integrals, and they can be analyzed with the same techniques.  
In this section we will show that one can write the $L$-th spectral trace as an $(L-1)$-dimensional integral. This leads to explicit expressions for the cases $L=1,2$, 
which can be evaluated in closed form in many cases. 

\subsection{General considerations} By using the cyclic property of the trace, one finds that
\be
\tr \rho_{m,n}^L=\operatorname{Tr}\left((f(\mathsf{q})g(\mathsf{p}))^L\right), 
\ee
where 
\be
\label{fg-ex}
f(x)= |\mypsi{a+c}{cn}(x)|^2= \frac {\re^{2\pi hx}}{2\cosh(\pi\mathsf{b} x)}, \qquad g(x)= |\mypsi{a}{c}(x)|^2, 
\ee
and we have used (\ref{psi-cosh}). We first establish a result which allows to compute this type of traces, for general $f(x)$, $g(x)$. 

\begin{lemma}
 Let $f,g\in L^1(\mathbb{R})$. Then
 \begin{multline}\label{trN}
\operatorname{Tr}\left((f(\mathsf{q})g(\mathsf{p}))^L\right)\\=
L\int_{\mathbb{R}^L}\operatorname{d}\! x_1\dots\operatorname{d}\! x_L \delta(x_1+\dots+x_L)J_f(x_1,\dots,x_L)\prod_{j=1}^L\tilde g(x_{j+1}-x_j), 
\end{multline}
where $x_{L+1}\equiv x_1$,
\begin{equation}\label{aver}
J_f(x_1,\dots,x_L)\equiv \int_{\mathbb{R}}f(x_1+t)\dots f(x_L+t)\operatorname{d}\! t, 
\end{equation}
and
\begin{equation}
\tilde g(x)\equiv \int_{\mathbb{R}}g(t)\re^{-2\pi\im xt}\operatorname{d}\! t.
\end{equation}
 \end{lemma}
 \begin{proof}
 We start by calculating the integral kernel
 \begin{multline}
\langle x\vert g(\mathsf{p})\vert y\rangle=\int_{\mathbb{R}} \tilde g(t)\langle x\vert \re^{2\pi\im t\mathsf{p}}\vert y\rangle\operatorname{d}\! t=\int_{\mathbb{R}} \tilde g(t)\langle x+t\vert y\rangle\operatorname{d}\! t\\=\int_{\mathbb{R}} \tilde g(t)\delta( x+t-y)\operatorname{d}\! t
=\tilde g(y-x).
\end{multline}
Now, by using the convention  $x_{L+1}\equiv x_1$, we write
\begin{multline}
 \operatorname{Tr}\left((f(\mathsf{q})g(\mathsf{p}))^L\right)=
\int_{\mathbb{R}^L}\operatorname{d}\! x_1\dots\operatorname{d}\! x_L
 \prod_{j=1}^L\langle x_j\vert f(\mathsf{q})g(\mathsf{p})\vert x_{j+1}\rangle\\
 =
\int_{\mathbb{R}^L}\operatorname{d}\! x_1\dots\operatorname{d}\! x_L
 \prod_{j=1}^L f(x_j)\langle x_j\vert g(\mathsf{p})\vert x_{j+1}\rangle=
\int_{\mathbb{R}^L}\operatorname{d}\! x_1\dots\operatorname{d}\! x_L
 \prod_{j=1}^L f(x_j)\tilde g(x_{j+1}-x_j)\\
 =
L\int_{\mathbb{R}^{L+1}}\operatorname{d}\! t\operatorname{d}\! x_1\dots\operatorname{d}\! x_L
\delta(x_1+\dots +x_L-L t) \prod_{j=1}^L f(x_j)\tilde g(x_{j+1}-x_j)\\
 =
L\int_{\mathbb{R}^{L+1}}\operatorname{d}\! t\operatorname{d}\! x_1\dots\operatorname{d}\! x_L
\delta(x_1+\dots +x_L) \prod_{j=1}^L f(x_j+t)\tilde g(x_{j+1}-x_j)\\
=L\int_{\mathbb{R}^L}\operatorname{d}\! x_1\dots\operatorname{d}\! x_L \delta(x_1+\dots+x_L)J_f(x_1,\dots,x_L)\prod_{j=1}^L\tilde g(x_{j+1}-x_j).
\end{multline}
\end{proof}
In our case, the function $f(x)$ has the explicit expression obtained in (\ref{fg-ex}). We now evaluate the function $J_f(x_1,\dots,x_L)$ in this particular case. 

\begin{lemma}
 If
 \begin{equation}\label{fx}
f(x)\equiv \frac {\re^{2\pi hx}}{2\cosh(\pi\mathsf{b} x)},
\end{equation}
then the function $J_f(x_1,\dots,x_L)$ defined by the equation~\eqref{aver} has the following explicit form
\begin{equation}\label{Jx1xN}
J_f(x_1,\dots,x_L)=\frac{1}{2^L\mathsf{b}\sin\!\left(\pi L\left(\frac12-\frac{h}{\mathsf{b}}\right)\right)}\sum_{j=1}^L\prod_{k\ne j}\frac{\re^{2\pi h(x_k-x_j)}}{\sinh(\pi\mathsf{b}(x_k-x_j))}.
\end{equation}
\end{lemma}
\begin{proof}
 By using Lemma~2.1 of \cite{GK}, we have
 \begin{multline}\label{calJ}
J_f(x_1,\dots,x_L)=\int_{\mathbb{R}}F(t)\operatorname{d}\! t=\left(\int_{\mathbb{R}}-\int_{\mathbb{R}+\im\mathsf{b}^{-1}}\right)\frac{F(t)}{1-F(t+\im/\mathsf{b})/F(t)}\operatorname{d}\! t\\
=\frac{1}{1-\lambda^{2L}}\left(\int_{\mathbb{R}}-\int_{\mathbb{R}+\im\mathsf{b}^{-1}}\right)F(t)\operatorname{d}\! t=\frac{2\pi\im}{1-\lambda^{2L}}\sum_{0<\Im z<\mathsf{b}^{-1}}\operatorname{Res}_{t=z}F(t),
\end{multline}
where
\begin{equation}
F(t)=\prod_{j=1}^Lf(x_j+t),
\end{equation}
so that
\begin{equation}
\frac{F(t+\im/\mathsf{b})}{F(t)}=\lambda^{2L},\quad \lambda\equiv- \im \re^{\pi\im h/\mathsf{b}}.
\end{equation}
The contribution in \eqref{calJ} comes from $L$ poles of $F(t)$ at the points $t=z_j\equiv -x_j+\frac{\im}{2\mathsf{b}}$, $1\le j\le L$, with the residues
\begin{equation}
\operatorname{Res}_{t=z_j}F(t)=\frac{(\lambda/2)^L}{\pi\mathsf{b}}\prod_{k\ne j}\frac{\re^{2\pi h(x_k-x_j)}}{\sinh(\pi\mathsf{b}(x_k-x_j))}.
\end{equation}
which lead straightforwardly to \eqref{Jx1xN}.
\end{proof}

When $L=1$, the trace can be computed immediately by using Remark~\ref{rem1}, and one finds
\be
\label{strace}
\tr \rho_{m,n}= {1\over 2 \mathsf{b} \cos \left( {\pi h \over \mathsf{b}} \right)} {
\fadn\!\left(c_{\mathsf{b}}-2\im a\right)\fadn\!\left(c_{\mathsf{b}}-2\im c\right)\over \fadn\!\left(c_{\mathsf{b}} -2 \im (a+c)\right)}.
\ee
When $L=2$, the formula~\eqref{trN} with $f(x)$, $g(x)$ given in (\ref{fg-ex}), together with Remark~\ref{rem1}, gives rise to
 \begin{multline}
 \label{ist}
\operatorname{Tr}(\rho_{m,n}^{2})\\=\frac{\left|\fad\left(2\im(a+c)-c_{\mathsf{b}}\right)\right|^2}{2\mathsf{b}\sin\left(2\pi h/\mathsf{b}\right)}\int_{\mathbb{R}}\frac{\sinh(2\pi hs)\fadn(s+c_{\mathsf{b}}-2\im a)\fadn(s+c_{\mathsf{b}}-2\im c)}{\sinh(\pi\mathsf{b}s)\fadn(s-c_{\mathsf{b}}+2\im a)\fadn(s-c_{\mathsf{b}}+2\im c)}\operatorname{d}\! s\\
=\frac{\left|\fad\left(2\im(a+c)-c_{\mathsf{b}}\right)\right|^2}{2\mathsf{b}\sin\left(2\pi h/\mathsf{b}\right)}\int_{\mathbb{R}}\frac{\sinh(2\pi hs)}{\sinh(\pi\mathsf{b}s)}W_{\frac{\myh}2-a}(s)W_{\frac{\myh}2-c}(s)
\operatorname{d}\! s, 
\end{multline}
where
\begin{equation}
W_a(x)\equiv\left|\mypsi{a}{a}(x)\right|^2.
\end{equation}
\subsection{Some explicit results} The spectral traces for $L=1, 2$ depend on $m,n$ and $\mathsf{b}$, and they can be computed very explicitly in some special cases. 
Using the results in \cite{GK}, one can evaluate the traces in closed form when $\mathsf{b}$ is of the form 
\be
\mathsf{b}={\sqrt{ M/N}}, 
\ee
and $M$, $N$ are coprime positive integers. Note that, due to (\ref{b-hbar}), this corresponds to rational values of $\hbar/\pi\in \BQ_{>0}$. Let us present some results in the case 
$m=n=1$, which is the operator associated to local $\IP^2$. In this case, the formula~\eqref{strace} for the trace can be written as 
\be
\label{gen-b}
\tr \rho_{1,1}= \exp \left[ -{\pi \im \over 36} \left( 12 c_{\mathsf{b}}^2 + 4 \mathsf{b}^2-3\right) \right] {\fad\left( c_{\mathsf{b}}- {\im \mathsf{b}\over 3} \right)^2  \over \mathsf{b}\sqrt{3}\fad\left( c_{\mathsf{b}}-  {2 \im \mathsf{b}\over 3} \right)}={1\over \mathsf{b}}\left|\fad\left( c_{\mathsf{b}}- {\im \mathsf{b}\over 3} \right)\right|^3.
\ee
By using Theorem~1.9 of \cite{GK} we obtain, for example,
\be
\label{first-trace}
\ba
\tr \rho_{1,1}\left(\hbar ={2 \pi \over 3} \right)&= {1\over 3} \exp \left( {V \over 2 \pi} \right), \qquad V= 2 \, {\rm Im}\left( {\rm Li}_2 \left(\re^{\pi \ri \over 3} \right)  \right), \\
\tr \rho_{1,1}\left(\hbar =\pi  \right)&= {1\over 2 {\sqrt{3}}}, \\
\tr \rho_{1,1}\left(\hbar =4 \pi  \right)&= {1\over 36}.
\ea
\ee
Note that $V$ is the volume of the figure-eight knot. 
\begin{remark}
 One can show that if 
 \begin{equation}
\mathsf{b}^2= {k+1 \over l+{1\over2}-{\varepsilon\over 6}},\qquad  k,l\in\mathbb{Z}_{\ge0},\quad \varepsilon\in\{\pm1\},
\quad \gcd\left(k+1,3l+{3-\varepsilon\over2}\right)=1,
\end{equation}
then formula \eqref{gen-b} can be simplified to
 \begin{equation}
\tr \rho_{1,1}={1\over\mathsf{b}}
\left|{\left(q_-\right)_{k+{1-\varepsilon\over2}}\over\mathsf{b}\left(q_+\right)_l}\right|^{3\varepsilon},\qquad q_\pm\equiv e^{2\pi\im\mathsf{b}^{\pm2}},\quad (x)_j\equiv (1-x)(1-x^2)\dots (1-x^j).
\end{equation}
\end{remark}
Let us now consider the case in which $m$ is arbitrary, but $n=1$. We further specialize the answer to the case $\hbar=2 \pi$. As we will shortly review, this value leads to a simplified 
conjecture for the Fredholm determinant of the operator $\rho_{m,n}$, and therefore it is interesting to see if such a simplification occurs in the formulae obtained above. 
When $\hbar=2 \pi$, one has from (\ref{b-hbar}) that
\be
\label{msb}
\mathsf{b}^2=m+2. 
\ee
Using (\ref{eq.bshift}), one can see that for this value of $\mathsf{b}$, the function $|\mypsi{a}{c}(x)|^2$ indeed simplifies to an elementary function, 
and the trace can be evaluated in closed form to be 
\be
\tr \rho_{m,1} \left( \hbar =2 \pi \right) = {1\over 4(m+2) \sin \left( {\pi \over m+2} \right) \sin \left( {2 \pi \over m+2} \right)}. 
\ee
In particular, we find 
\be
\label{max-first}
\ba
\tr \rho_{1,1} \left( \hbar =2 \pi\right)&={1\over 9}, \\
\tr \rho_{2,1}\left( \hbar =2 \pi\right)&={1\over 8 {\sqrt{2}}}. 
\ea
\ee

Let us now list some useful values for the second spectral trace. Like before, the integrand in (\ref{ist}) simplifies to an elementary 
function when $\mathsf{b}$ satisfies (\ref{msb}), and the resulting integral can be calculated straightforwardly. One finds in this way,  
\be
\label{max-second}
\ba
\tr \rho^2_{1,1} \left( \hbar =2 \pi\right)&={1\over 27}-{1\over 6 {\sqrt{3}} \pi}, \\
\tr \rho^2_{2,1}\left( \hbar =2 \pi\right)&={1\over 16 \pi}-{1\over 64}. 
\ea
\ee

\subsection{Refined traces} As we have seen in subsection \ref{sym-eqs}, in the case of the operator $\rho_{1,1}$, there is a $\mathbb{Z}_3$ symmetry implemented by the unitary operator 
$\mathsf{V}$ of order three, defined in (\ref{Vdef}), which commutes with $\rho_{1,1}$. This means that both operators can be diagonalized simultaneously. 
The subspace spanned by eigenfunctions with a fixed eigenvalue of $\mathsf{V}$ can be obtained by using the projection operators, 
\begin{equation}
\mathsf{P}_i\equiv\frac13\sum_{j=0}^2\omega^{-ij}\mathsf{V}^j, \quad i=0,1,2,\quad \omega\equiv \re^{2\pi\im/3},
\end{equation}
which commute with $\rho_{1,1}$ and solve the spectral problem for $\mathsf{V}$, in the sense that they satisfy the relations
\begin{equation}
\mathsf{P}_i\mathsf{P}_j=\delta_{i,j}\mathsf{P}_i,\quad \mathsf{V}\mathsf{P}_i=\omega^i\mathsf{P}_i,\quad \sum_{j=0}^2\mathsf{P}_j=1.
\end{equation}
In particular, we can consider the refined spectral traces 
\be
\label{rst}
\operatorname{Tr}\left(\rho_{1,1}^{L}\mathsf{P}_i\right), \qquad i=0,1,2, 
\ee
which can be calculated as $(L+1)$-dimensional integrals by using the explicit expressions for the kernels of $\rho_{1,1}$ and $\mathsf{V}$ (in fact, one can 
simplify them to $L$-dimensional integrals. A detailed treatment of these refined traces will appear elsewhere.) 

Empirically, one finds that the spectrum of $\rho_{1,1}$ is non-degenerate, with eigenfunctions $\{ |\varphi_n \rangle \}_{n\in \mathbb{Z}_{\ge 0}}$. We will write,  
\be
\rho_{1,1} |\varphi_n \rangle = \re^{-E_n}|\varphi_n \rangle, \qquad  n \in \mathbb{Z}_{\ge 0}, 
\ee
where we order the eigenvalues as 
\be
E_0 < E_1< E_2 < \cdots
\ee
Due to the $\mathbb{Z}_3$ symmetry implemented by $\mathsf{V}$, each of the 
eigenfunctions of $\rho_{1,1}$ is 
also an eigenfunction of $\mathsf{V}$, $i=0, 1,2$, and one finds that
\be
 \mathsf{V}  |\varphi_{i+3m} \rangle= \omega^i  |\varphi_{i+3m} \rangle, \qquad m \in \mathbb{Z}_{\ge 0}, \qquad i=0,1,2. 
 \ee
 Such a $\mathbb{Z}_3$ decomposition of the spectrum was noted in the numerical analysis of the spectrum of $\rho_{1,1}$ in \cite{GHM}. 
The refined spectral traces (\ref{rst}) are then given by, 
\be
\operatorname{Tr}\left(\rho_{1,1}^{L}\mathsf{P}_i\right)= \sum_{m \ge 0} \re^{-L E_{i+3 m }}.
\ee
Note that the structure of the spectrum in this problem is similar to what happens in bound-state problems in ordinary one-dimensional Quantum Mechanics, in the presence of a $\mathbb{Z}_2$ 
parity symmetry: in that case, the eigenstates of the Hamiltonian $|\varphi_n\rangle$, $n \in \mathbb{Z}_{\ge 0}$, with energies $E_0<E_1< E_2 < \cdots$, have parity $(-1)^n$.

\section{Comparison with the conjecture of \cite{GHM}} In the previous sections, we have shown that the operators $\rho_S= \mathsf{O}_S^{-1}$ arising from mirror curves, for many 
toric del Pezzo threefolds, are of trace class. This means that their Fredholm determinant 
\be
\Xi_S(\kappa, \hbar)={\rm det} \left( 1+ \kappa \rho_S\right) 
\ee
exists and is an entire function of $\kappa$ \cite{Si}. The Taylor expansion of the Fredholm determinant around $\kappa=0$ provides a generating 
functional of the spectral traces, since 
\be
\log \Xi_S(\kappa, \hbar)= -\sum_{L=1}^\infty {(-\kappa)^L \over L} \tr \rho_S^L. 
\ee
According to the conjecture of \cite{GHM}, the Fredholm determinant of the operators $\rho_S$ can be computed exactly in terms of enumerative geometry of $X$. This provides 
as well an exact semiclassical expansion for the spectrum of the operators $\rho_S$, in terms of an exact quantization condition akin to those found in 
conventional Quantum Mechanics (see, for example, \cite{ZJJ}). The perturbative, WKB part of this quantization condition agrees with the result of \cite{MM, ACDKV}, but it contains 
in addition instanton corrections which are determined by the standard Gromov--Witten (or Gopakumar--Vafa (GV) invariants \cite{GV}) of $X$. 

A detailed statement of the conjecture can be found in \cite{GHM}. We provide here a short summary, and we state the result for local $\IP^2$. This 
conjecture involves two generating functionals of enumerative invariants of a toric Calabi--Yau threefold $X$: the Nekrasov--Shatashvili (NS) limit \cite{NS} 
of the refined topological string free energy \cite{IKV,CKK,NO}, which we will call for simplicity the NS free energy, 
and the standard topological string free energy in the GV representation. The latter is given by 

\be
\label{gv-form}
\mathsf{F}^{\rm WS}\left({\bf t}, g_s\right)= \sum_{g\ge 0} \sum_{\bf d} \sum_{w=1}^\infty {1\over w} n_g^{ {\bf d}} \left(2 \sin { w g_s \over 2} \right)^{2g-2} \re^{-w {\bf d} \cdot {\bf t}}. 
\ee
In this expression, the vector ${\bf t}$ contains the K\"ahler parameters of the Calabi--Yau $X$, ${\bf d}$ is the vector of degrees (which are non-negative integers), $g_s$ is a parameter 
usually called the string coupling constant, and $n_g^{\bf d}$ are integer invariants called GV invariants of $X$. They depend on the degree ${\bf d}$ and on the genus $g$. The NS free energy 
is given by 
\be
\label{NS-j}
\mathsf{F}^{\rm NS}({\bf t}, \hbar) =\sum_{j_L, j_R} \sum_{w, {\bf d} } 
N^{{\bf d}}_{j_L, j_R}  \frac{\sin\frac{\hbar w}{2}(2j_L+1)\sin\frac{\hbar w}{2}(2j_R+1)}{2 w^2 \sin^3\frac{\hbar w}{2}} \re^{-w {\bf d}\cdot{\bf  t}}. 
\ee
Here, $N^{{\bf d}}_{j_L, j_R}$ are the refined BPS invariants of $X$ considered \cite{IKV,CKK,NO}. They depend on the degree ${\bf d}$ and 
on two half-integers (or ``spins") $j_L$, $j_R$. The GV invariants $n_g^{\bf d}$ can be written as particular combinations of the refined BPS invariants. 

According to the conjecture of \cite{GHM}, the Fredholm determinant of the operators $\rho_S$ associated to 
the mirror curves can be expressed in terms of these two generating functionals. We will state the result for local $\IP^2$, which has one single K\"ahler parameter $t$. We first need a dictionary 
relating the parameter $\kappa$ appearing in the Fredholm determinant, to the geometric data of local $\IP^2$. To do this, we recall that local $\IP^2$ has a single complex parameter or modulus $\tilde u$, and we identify 
\be
\kappa \equiv\re^\mu= \tilde u,
\ee
where in the second equation we have introduced the ``fugacity" $\mu$. The complex modulus $\tilde u$ is related to the K\"ahler parameter $t$ through the so-called mirror map. 
As shown in \cite{ACDKV}, the WKB approach makes it possible to define a ``quantum" mirror map $t(\hbar)$, which in the limit $\hbar \rightarrow 0$ agrees with the conventional 
mirror map. The quantum mirror map of local $\IP^2$ can be computed as a power series in $\re^{-3 \mu}$, 
with $\hbar$-dependent coefficients, and it has been studied in \cite{ACDKV,HKRS}. One finds, for the very first orders, 
\be
\label{muef-def}
\mu_{\rm eff}\equiv {t(\hbar)\over 3}= \mu+ 2 \cos\left( {\hbar \over 2} \right)  {\rm e}^{-3 \mu}+ \cdots,
\ee
which defines the ``effective" chemical potential $\mu_{\rm eff}$. 

We now introduce three different functions of $\mu$. The first one is the ``perturbative" grand potential,  
 \be
 \label{jp}
 \mathsf{J}^{({\rm p})}(\mu, \hbar)= {C(\hbar)\over 3} \mu^3 +B(\hbar) \mu + A(\hbar), 
 \ee
 where 
 \be
 \label{bc}
 C(\hbar)=  {9 \over 4 \pi \hbar}, \qquad B(\hbar)= {\pi \over 2 \hbar}-{\hbar \over 16 \pi}. 
 \ee
 The function $A(\hbar)$ is more complicated, and it is defined as follows. We first define,
 \be
\label{ak}
A_{\rm c}(k)= \frac{2\zeta(3)}{\pi^2 k}\left(1-\frac{k^3}{16}\right)
+\frac{k^2}{\pi^2} \int_0^\infty \frac{x}{\re^{k x}-1}\log(1-\re^{-2x})\rd x.
\ee
This function was first introduced, in a very different context, in \cite{MP}, and determined in integral form in \cite{HHHNSY,HO}. It can be obtained by an appropriate 
resummation of the ``constant map contribution" appearing in Gromov--Witten theory \cite{BCOV,MaMo,FP}. 
Then, the function $A(\hbar)$ in (\ref{jp}) is given by 
\be
\label{ah-p2}
A(\hbar) ={3 A_{\rm c}(\hbar/\pi)- A_{\rm c}(3\hbar/\pi) \over 4}. 
\ee
The two other functions of $\mu$ are obtained from the generating functionals of enumerative invariants introduced before. 
The first one is the ``membrane" grand potential, 
\be
\mathsf{J}^{\rm M2}(\mu_{\rm eff}, \hbar)= {t(\hbar) \over 2 \pi} {\partial \mathsf{F}^{\rm NS}(t(\hbar), \hbar) \over \partial t} 
+{\hbar^2 \over 2 \pi} {\partial \over \partial \hbar} \left(  {\mathsf{F}^{\rm NS}(t(\hbar), \hbar) \over \hbar} \right),
\ee
where $t(\hbar)=3 \mu_{\rm eff}$, as in (\ref{muef-def}). The second one is the ``worldsheet" grand potential
\be
\label{ws-p2}
\mathsf{J}^{\rm WS}(\mu_{\rm eff}, \hbar)= \mathsf{F}^{\rm WS}\left( {6 \pi \mu_{\rm eff} \over \hbar}+ \pi \ri , {4 \pi^2 \over \hbar} \right).
\ee
The total grand potential is the sum of these three functions, 
\be
\label{jtotal}
\mathsf{J}_{\IP^2}(\mu, \hbar) = \mathsf{J}^{({\rm p})}(\mu_{\rm eff}, \hbar)+ \mathsf{J}^{\rm M2}(\mu_{\rm eff}, \hbar)+ \mathsf{J}^{\rm WS}(\mu_{\rm eff}, \hbar), 
\ee
and it was first considered in \cite{HMMO}. 

According to the conjecture of \cite{GHM}, the Fredholm determinant of the operator $\rho_{\IP^2}=\rho_{1,1}$ is given by 
\be
\label{spec-det}
\Xi_{\IP^2} (\kappa, \hbar)= \sum_{n \in \BZ} \exp\left(  \mathsf{J}_{\IP^2}(\mu + 2 \pi \im n, \hbar) \right).  
\ee
The sum over $n$ defines a ``generalized theta function" $\Theta_{\IP^2} \left( \mu, \hbar \right)$, 
\be
\label{gen-theta}
\Xi_{\IP^2} (\kappa, \hbar)=\re^{ \mathsf{J}_{\IP^2}(\mu, \hbar) }  \Theta_{\IP^2} \left( \mu, \hbar \right). 
\ee
The expression (\ref{spec-det}) is somewhat formal, since the convergence properties of the r.h.s. have not been established. In particular, the grand potential 
is defined in (\ref{jtotal}) as a formal power series in the two exponentials
\be
\re^{-3 \mu}, \qquad \re^{-6 \pi \mu/\hbar}. 
\ee
Note that both formal power series, $\mathsf{J}^{\rm M2}(\mu_{\rm eff}, \hbar)$ and $\mathsf{J}^{\rm WS}(\mu_{\rm eff}, \hbar)$, have poles 
when $\hbar/\pi\in \BQ_{>0}$. As explained in \cite{KM}, this 
is indeed the reason why the WKB expansion of \cite{ACDKV} can not account for the spectral properties of the operator $\rho_{\IP^2}$. However, it is easy 
to show that in the total grand potential (\ref{jtotal}) these poles cancel \cite{HMMO} (this property was discovered, in a related example, in \cite{HMO}), 
so that $\mathsf{J}_{\IP^2}(\mu, \hbar)$ is well-defined as a formal power series. There is in 
addition evidence from explicit computations that the grand potential $\mathsf{J}_{\IP^2}(\mu, \hbar)$ is 
analytic in a neighbourhood of infinity in the $\mu$ plane, at least when 
$\hbar$ is real. 

As noted in \cite{GHM}, the value $\hbar=2 \pi$ is special. In that case, the grand potential $\mathsf{J}_{\IP^2}(\mu, 2 \pi)$ 
simplifies and one finds \cite{GHM}
\be
\label{j-p2}
\mathsf{J}_{\IP^2}(\mu, 2\pi)= A(2 \pi)+
  {1\over 4 \pi^2} \left( \widehat F_0(t) - t \partial_t  \widehat F_0(t) + {t^2 \over 2} \partial_t^2  \widehat F_0(t)\right)+  \widehat F_1(t)+ \widehat F_1^{\rm NS}(t). 
\ee
In this equation, $ \widehat F_0(t)$, $ \widehat F_1(t)$ are the standard genus zero and genus one free energies of local $\IP^2$, but after setting $\re^{-t} \rightarrow -\re^{-t}$ in the 
worldhsheet instanton expansion (this is due to the shift in (\ref{ws-p2})). $\widehat F_1^{\rm NS}(t)$ is obtained by expanding the 
NS free energy up to next-to-leading order in $\hbar$, and performing the same change of sign. 
The K\"ahler parameter $t$ is now related to $\mu$ by the standard mirror map, again up to a sign (see \cite{GHM} for further details of this computation and explicit 
expressions for these free energies). In addition, one has that
\be
A(2 \pi)= {1\over 6} \log(3)- {\zeta(3) \over 3 \pi^2}. 
\ee
 From the expression (\ref{j-p2}), it is easy to see that $\mathsf{J}_{\IP^2}(\mu, 2\pi)$ is analytic in the $\mu$ plane, with a region of analyticity 
determined by the standard conifold singularity of local $\IP^2$. For this special value of $\hbar$, it is also 
possible to show that the generalized theta function appearing in (\ref{gen-theta}) becomes a standard Jacobi theta function:
\be
 \Theta_{\IP^2} \left( \mu, 2 \pi \right)=\vartheta_3 \left( \xi-{3 \over 8}, {9 \tau \over 4}\right), 
 \ee
 where
 \be
 \label{xi-f}
 \xi= {3\over 4 \pi^2} \left( t \partial_t^2 \widehat  F_0(t) -\partial_t  \widehat F_0(t)\right),  \qquad 
 \tau={2 \ri \over \pi} \partial_t^2  \widehat  F_0(t). 
 \ee
Note that the $\tau$ appearing here is the standard modulus of the genus one mirror curve of local $\IP^2$. In particular, one has that ${\rm Im}(\tau)>0$. 
Therefore, in the case $\hbar=2 \pi$, the formula for 
the Fredholm determinant (\ref{spec-det}) is completely explicit and well-defined. It can be expanded around $\kappa=0$ by using the analytic continuation of the various quantities appearing 
here. One finds (see \cite{GHM}, eq. (4.56)),
\be
\Xi_{\IP^2} (\kappa, 2 \pi)=1 + {\kappa \over 9} + \left( {1\over 12 {\sqrt{3}} \pi}-{1\over 81} \right) \kappa^2 + \CO(\kappa^3), 
\ee
in perfect agreement with the results (\ref{max-first}), (\ref{max-second}). 

What happens when $\hbar \not= 2 \pi$? In this case, it is difficult to extract analytic results for the spectral traces from (\ref{spec-det}), but there is a good 
geometric reason for that: the grand potential is expressed as a formal power 
series around the so-called large radius point at $\mu \rightarrow \infty$, while the spectral traces are obtained by 
expanding around $\kappa=0$, i.e. around $\mu= -\infty$. This is the orbifold point of the geometry, and there are no known resummations 
of the generating functions of BPS invariants 
around that point. However, one can still obtain very precise {\it numerical} results for the spectral traces. In this numerical calculation, the 
precision increases with the number of terms retained in the expansion at large radius. This makes it possible to test the conjecture (\ref{spec-det}) 
against the analytical results for the spectral traces derived in this paper, with arbitrary precision. 

The numerical calculation of the traces for general $\hbar$ proceeds as follows. Let us write down the expansion of $\Xi_{\IP^2} (\kappa, \hbar)$ around $\kappa=0$ as 
\be
\Xi_{\IP^2}(\kappa, \hbar)=1+\sum_{N=1}^\infty Z_{\IP^2}(N, \hbar) \kappa^N. 
\ee
We will call $Z_{\IP^2}(N, \hbar)$ the canonical partition function. It can be written in terms of the spectral traces $\tr \rho_{1,1}^{N'}$ with $N'\le N$, and viceversa. For example, 
one has 
\be
\ba
Z_{\IP^2}(1, \hbar) &=\tr \rho_{1,1} , \\
Z_{\IP^2}(2, \hbar)&= {1\over 2} \left( \left(  \tr \rho_{1,1} \right)^2 -\tr \rho^2 _{1,1}   \right). 
\ea
\ee
By writing $Z_{\IP^2}(N, \hbar)$ as a 
contour integral and performing a contour deformation \cite{HMO,GHM}, one obtains the following integral representation of the partition function,  
\be
\label{int-Z}
Z_{\IP^2}(N, \hbar) = {1\over 2 \pi \ri} \int_{\mathcal C} \re^{J_{\IP^2}(\mu, \hbar)- N \mu} \rd \mu, 
\ee
where ${\mathcal C}$ is a contour going from $\re^{-\im \pi/3} \infty$ to $\re^{\im \pi/3} \infty$. In deriving this representation, 
we assumed that $J_{\IP^2}(\mu, \hbar)$ is analytic in a neighborhood of infinity. Note that ${\mathcal C}$ is the standard contour for 
the integral representation of the Airy function. This contour is appropriate due to the fact that $J_{\IP^2}(\mu, \hbar)$ is a 
cubic polynomial, plus exponentially small corrections. By expanding the exponentially small corrections in $J_{\IP^2}(\mu, \hbar)$, 
the partition functions $Z_{\IP^2}(N, \hbar)$ can be evaluated as an infinite sum of Airy functions and their derivatives. More concretely, if we write
\begin{equation}
\re^{J_{\IP^2}(\mu, \hbar)} = \re^{J^{({\rm p})} (\mu,  \hbar)} \sum_{l, n} a_{l,n} \mu^n \re^{- l \mu}, 
\end{equation}
we find
\begin{equation}
\label{z-exp}
Z_{\IP^2}(N, \hbar) =\frac{\re^{A(\hbar)}}{\left(C(\hbar) \right)^{1/3}}
\sum_{l,n} a_{l,n} \left(-\frac{\partial}{\partial N}\right)^n \mathrm{Ai}
\left(\frac{N+l -B(\hbar)}{\left(C(\hbar)\right)^{1/3}}\right), 
\end{equation}
where ${\rm Ai}(z)$ is the Airy function. Note that $n$ takes non-negative integer values, but $l$ is of the form $3 p + 6 \pi q/\hbar$, with $p,q$ non-negative integers. 
The infinite sum in (\ref{z-exp}) turns out to converge very rapidly, and produces highly accurate numerical answers for the partition functions. 
Using this procedure, one obtains for example, 
\be
Z_{\IP^2}\left(1, \hbar={2 \pi \over 3} \right)= 0.4604521481728325977904889856168747087632124207...
\ee
where we have kept the digits which are stable as we increase the number of exponentially small corrections in the formal power series defining $\mathsf{J}_{\IP^2}(\mu, \hbar)$. 
This is in agreement with the analytic result in the first line of (\ref{first-trace}) (by pushing the calculation of $\mathsf{J}_{\IP^2}(\mu, \hbar)$ up to degree $d=12$, we can obtain 
an agreement with the analytic result of 60 decimal digits). We can test with similar accuracy the other analytic results presented in (\ref{first-trace}) for local $\IP^2$. 
It is also possible to test the results for the first two spectral traces of the operator $\rho_{2,1}$ in (\ref{max-first}), (\ref{max-second}), 
by using the conjecture of \cite{GHM}, as applied to local $\IF_2$ in \cite{GKMR}.

In this paper we have assumed that $\hbar$ is real and positive, but it is clear that many of our results can be extended to complex values of $\hbar$.  
For example, the formulae we have obtained for the spectral traces can be continued to complex $\hbar$, provided that ${\rm Im}(\hbar)>0$ (this follows from the relationship 
(\ref{b-hbar}) and the constraint for $\mathsf{b}^2$  in (\ref{qq})). The conjectural expression for the Fredholm determinant (\ref{spec-det}) can be also extended to complex values of $\hbar$, and 
a numerical calculation of the traces is indeed possible, by using the formula (\ref{z-exp}). We have explicitly verified that the result (\ref{gen-b}) for the first spectral 
trace agrees with the numerical calculation of $Z_{\IP^2}\left(1, \hbar \right)$ for various complex values of $\hbar$. This indicates that the conjecture of \cite{GHM} has a natural 
extension to complex $\hbar$, which agrees with the analytic results obtained in this paper. 

 \section{Conclusions and outlook}

In this paper we have shown that, for a large class of Calabi--Yau geometries $X$ associated to toric, almost del Pezzo surfaces $S$ by (\ref{dP}), 
the quantization of mirror curves proposed in \cite{GHM} leads to a new family of 
trace class, positive-definite operators $\rho_S$ on $L^2(\BR)$. According to the conjecture put forward in \cite{GHM}, the spectral problem 
associated to these operators is exactly solvable: the spectrum is determined by an exact quantization condition, derived in turn from an exact expression 
for their Fredholm determinant. Our results for the spectral traces in this paper lead to very explicit and non-trivial tests of the conjecture of \cite{GHM}. 
It is clear that this new family of operators deserves further study, as they provide a fascinating bridge between the spectral 
theory of trace class operators, and the enumerative geometry of toric Calabi--Yau threefolds. Although we have restricted ourselves to geometries of the form (\ref{dP}), 
similar considerations should apply to more general toric Calabi--Yau threefolds.

Our results also have some conceptual consequences for topological string theory. Indeed, one of the motivations 
of \cite{GHM} is to provide a non-perturbative approach to topological strings in the case of toric Calabi--Yau manifolds. 
In order to have a non-perturbative approach, we need a well-defined problem or quantity. Thanks to our results, it is clear that the Fredholm determinant of the operator associated to 
the mirror curve is such a quantity. As explained in \cite{GHM}, the conventional topological string free energy appears as a 't Hooft expansion of the 
logarithm of the Fredholm determinant, and therefore it emerges as an asymptotic approximation to a well-defined quantity. 

An important aspect of our proof is its constructive character, in the sense that we provide an explicit expression 
for the integral kernel of the operators $\rho_S$ 
in some cases, as well as explicit expressions for some spectral traces. As we have shown, this kernel involves in an essential way Faddeev's quantum dilogarithm, 
and the resulting expressions for the spectral 
traces are very similar to the state-integral invariants of three-manifolds studied in \cite{Hi,DGLZ,AK,KLV,DG,AK:complex,Dim}. It is particularly intriguing that quantities like the volume of the $4_1$ knot appear in this context, 
as we have seen in (\ref{first-trace}). This might be an indication of a deeper relation between complex Chern--Simons theory and topological string theory. 

The explicit expressions that we have obtained for the integral kernels of the three-term operators 
make it possible to write down matrix integral expressions for the corresponding canonical partition functions (similar to (\ref{trN})). It would 
be interesting to study these matrix integrals with large $N$ techniques, and verify in this way some of the statements in \cite{GHM}. In view of the similarity between 
the integral kernel (\ref{ex-k}) and the kernels considered in \cite{Z,TW}, it would be also interesting to see if the Fredholm determinant of the operator $\rho_S$ is encoded 
in a TBA system. 

In the case of $\rho_{1,1}$, the operator corresponding to local $\IP^2$, we have found the explicit unitary operator which implements the $\mathbb{Z}_3$ symmetry of the spectrum, 
and we have seen that one can define refined spectral traces which take into account the decomposition of the Hilbert space w.r.t. the global symmetry. Similarly, one can 
define refined Fredholm determinants. It would be very interesting to find a refined version of the conjecture in \cite{GHM}, giving an explicit 
expression for these refined determinants. In the case of 
quantum-mechanical operators with parity symmetry, the corresponding refined Fredholm determinants satisfy 
sometimes non-trivial functional relations which determine the spectrum (see for example \cite{DT} for a review in the case of conventional 
Quantum Mechanics, and \cite{GHMb} for an example closely related to the family of operators studied here). This suggests that a more 
precise understanding of the refined spectral traces, 
in the case of $\rho_{1,1}$, might shed light on the solvability of the spectral problem. 

Finally, we hope that our results for the operators associated to mirror curves 
will suggest a strategy to prove the conjecture of \cite{GHM}.


\appendix

\section{Some useful properties of the quantum dilogarithm}
\lbl{sec.QDL}

The quantum dilogarithm $\fad(x)$ is defined by 
\cite{Faddeev}
\begin{equation}\lbl{fad}
\fad(x)
=\frac{(\re^{2 \pi \mathsf{b} (x+c_{\mathsf{b}})};q)_\infty}{
(\re^{2 \pi \mathsf{b}^{-1} (x-c_{\mathsf{b}})};\tq)_\infty} \,,
\end{equation}
where
\be
\label{qq}
q=\re^{2 \pi \im \mathsf{b}^2}, \qquad 
\tq=\re^{-2 \pi \im \mathsf{b}^{-2}}, \qquad 
\mathrm{Im}(\mathsf{b}^2) >0. 
\ee
An integral representation in the strip $|\mathrm{Im} z| < |\mathrm{Im} \, c_{\mathsf{b}}|$ is given by 
\be
\fad(x)=\exp \left( \int_{\mathbb{R}+\im\epsilon}
\frac{\re^{-2\im xz}}{4\sinh(z\mathsf{b})\sinh(z\mathsf{b}^{-1})}
{\operatorname{d}\!z  \over z} \right).
\ee
Remarkably, this function admits 
an extension to all values of $\mathsf{b}$ with 
$\mathsf{b}^2\not\in\mathbb{R}_{\le 0}$. $\fad(x)$ is 
a meromorphic function of $x$ with
\be
\text{poles:} \,\,\, c_{\mathsf{b}} + \im \BN \mathsf{b} + \im \BN \mathsf{b}^{-1},
\qquad
\text{zeros:} \,\, -c_{\mathsf{b}} - \im \BN \mathsf{b} - \im\BN \mathsf{b}^{-1} \,.
\ee
The functional equation
\be
\fad(x) \fad(-x)=\re^{\pi \im x^2} \fad(0)^2, 
\qquad
\fad(0)=\left(\frac{q}{\tilde q}\right)^{\frac{1}{48}} =\re^{\pi\im\left(\mathsf{b}^2+\mathsf{b}^{-2}\right)/24}
\ee
allows us to move $\fad(x)$ from the denominator to the numerator. In addition, when $\mathsf{b}$ is either real or on the unit circle, we have 
the unitarity relation
\be
\label{unit-fad}
{\overline{\fad(x)}}={1\over \fad\left( \overline x\right)}. 
\ee

The asymptotics of the 
quantum dilogarithm are given by~\cite[App.A]{AK}
\begin{equation}
\lbl{eq.as}
\fad(x) \sim \begin{cases}
\fad(0)^2\re^{\pi \im x^2} & \text{when} \quad \Re(x) \gg 0, \\
1             & \text{when} \quad \Re(x) \ll 0.
\end{cases}
\end{equation}

The quantum dilogarithm is a quasi-periodic function. Explicitly,
it satisfies the equations
\begin{subequations}
\begin{align}
\lbl{eq.bshift}
\frac{\fad(x+c_{\mathsf{b}}+\im \mathsf{b})}{
\fad(x+c_{\mathsf{b}})} 
&= \frac{1}{1-q \re^{2 \pi \mathsf{b} x}} 
\\ 
\lbl{eq.tbshift}
\frac{\fad(x+c_{\mathsf{b}}+\im\mathsf{b}^{-1})}{
\fad(x+c_{\mathsf{b}})} &= 
\frac{1}{1-\tq^{-1} \re^{2 \pi \mathsf{b}^{-1} x}} \,.
\end{align}
\end{subequations}

%
%


\bibliographystyle{hamsalpha}
\bibliography{biblio}
\end{document}